\documentclass[default,iicol]{sn-jnl}

\jyear{2021}%

\usepackage[caption=false]{subfig}

\usepackage{hyperref}
\usepackage{amsmath}
\usepackage{amsthm}
\usepackage{amssymb}
\usepackage{graphics} 
\usepackage{epstopdf} 
\usepackage{epsfig} 
\usepackage{color,soul}
\usepackage{calc}
\usepackage[numbers,sort&compress]{natbib}
\bibpunct[, ]{[}{]}{,}{n}{,}{,}
\makeatletter
\def\NAT@def@citea{\def\@citea{\NAT@separator}}
\makeatother

\theoremstyle{thmstyleone}%
\newtheorem{theorem}{Theorem}
\newtheorem{definition}{Definition}
\newtheorem{lemma}{Lemma}
\newtheorem{proposition}{Proposition}
\newtheorem{assumption}{Assumption}

\newtheorem{remark}{Remark}
\newtheorem{corollary}{Corollary}
\newtheorem{scenario}{Scenario}

\begin{document}

\title[Spoof-resilient Distributed Observers]{Cooperative Distributed State Estimation: Resilient Topologies against Smart Spoofers}


\author*[1]{\fnm{Mostafa} \sur{Safi}}\email{halebi@aut.ac.ir}

\affil*[1]{\orgdiv{} \orgname{Amirkabir University of Technology}, \orgaddress{\street{Hafez}, \city{Tehran}, \postcode{424}, \state{Tehran}, \country{Iran}}}

%


\abstract{A network of observers is considered, where through asynchronous (with bounded delay) communications, they cooperatively estimate the states of a Linear Time-Invariant (LTI) system. In such a setting, a new type of adversary might affect the observation process by impersonating the identity of the regular node, which is a violation of communication authenticity. These adversaries also inherit the capabilities of Byzantine nodes, making them more powerful threats called \textit{smart spoofers}. We show how asynchronous networks are vulnerable to smart spoofing attack. In the estimation scheme considered in this paper, information flows from the sets of source nodes, which can detect a portion of the state variables each, to the other follower nodes. The regular nodes, to avoid being misguided by the threats, distributively filter the extreme values received from the nodes in their neighborhood. Topological conditions based on strong robustness are proposed to guarantee the convergence. Two simulation scenarios are provided to verify the results.}




\maketitle

\section{Introduction}\label{Sect: Intro}

Security is becoming an increasingly important concern for the stability and safety of networked control systems. Nowadays, in large-scale control systems,
communication channels connecting various physical components for real-time measurement and control mostly make use of general purpose cyber-networks such as the Internet and wireless networks, which create vulnerabilities to adversarial intrusions. 
While conventional network security-based measures may be partially effective, novel resiliency methods explicitly taking the dynamical nature of physical components into account should be developed as any failure in security of the cyber components in such systems may turn into irrecoverable harms to the physical infrastructure.

Security experts define various security goals including (i) \textit{Confidentiality}, ensuring privacy of important data against outside eavesdroppers; (ii) \textit{Integrity}, maintaining fidelity of system signals; (iii) \textit{Availability}, capability of timely having access to the required signals; (iv) \textit{Authenticity}, verifying identity of each signal; (v) \textit{Authorization}, adjusting legitimacy of access by each component to other parts of the system; and (vi) \textit{Accountability}, detection of any potential attacks and faults in the system  \cite{daswani2007foundations}. 

In this paper, we consider \textit{masquerading}, \textit{spoofing}, or \textit{impersonation} attack strategy on cyber-physical networked systems, which is a threat of authenticity. A broad range of wired and wireless networks including sensor networks, in-vehicle networks, and Internet-based networks are susceptible to be threatened by spoofing. For instance, the reader can refer to \cite{zheng2012design} for satellite mobile communication networks, \cite{villalba2011secure} for mobile ad hoc networks, and \cite{ueda2015security} for CAN-based networks. Spoof-resiliency techniques would be essential for all of these network setups to detect and/or mitigate the adversarial effects. However, mostly in literature, the spoof-resilient algorithms are studied for the interactions between only two agents: a spoofer and a normal \cite{sun2012sack2,magiera2015detection,su2016stealthy,dutta2017confiscating,zhang2017strategic}. For example, \cite{magiera2015detection} presents an application of spatial processing methods for spoofing detection and mitigation. Also, a GPS spoofing scenario is formulated as a constrained optimization problem and an effective solution is provided to compute the falsified GPS measurement of each time instant \cite{su2016stealthy}. The false-data injection attack on unmanned vehicles is investigated in \cite{dutta2017confiscating}. Although, this differs a little from spoofing attacks. The attacker masquerades as a disturbance for control system of a vehicle and deviates its path smoothly. Furthermore, a game-theoretic approach is developed in \cite{zhang2017strategic} to counteract spoofing attacks. However, a common point all the above researches share is that there is no network of agents. Only two-side interplay scenarios are considered, where the spoofing or masquerading is the attacking method. Recently though, \cite{gil2017guaranteeing} and \cite{renganathan2017spoof} focused on the sequels of spoofing attack on the network of agents. However, both of these references use physical fingerprints of communication signals to undo the attacks, which is a different approach and cannot resist against onmniscient adversaries in practice. Despite \cite{gil2017guaranteeing}, in our work, the attackers do not leave any sign and thus the regular nodes cannot identify them. Also, omniscient attackers in our setup could break any type of signal encryption and perform masquerading. Particularly, our emphasis is on the resiliency of a network in terms of its topology that is a more basic level of counteraction to cyber  threats. Moreover, in \cite{gil2017guaranteeing}, attackers cause an availability threat by jamming the server with fake identities, which is a special case of our adversarial model. We combine adversarial capabilities of the so-called \textit{Byzantine} model, which is an integrity attack capable of sending inconsistent erroneous signals to the receivers introduced and used in \cite{lynch1996distributed,dibaji2019resilient}, with spoofing, that is use of other nodes' identities to send data on their behalf, and introduce a novel and more powerful adversarial model called \textit{smart spoofer}. In \cite{bonnet2016tight}, resiliency of synchronous networks is investigated against mobile Byzantine
adversaries that are different from our adversarial model. In our setting, smart spoofers can use the asynchrony of network communications to mislead the nodes with impersonated identities.

One of the targets of spoofers in network systems would be inserting erroneous values into the distributed state estimations performed by the nodes. Distributed state estimation algorithms are extensively studied in the literature \cite{alexandru2017limited,khan2010connectivity,park2017design,wang2017distributed}. However, all these research works focused on the interaction between dynamic system, observers and the graph topology. A  minimum cost communication graph which enables limited communication for decentralized estimation is investigated in \cite{alexandru2017limited} . The interplay between network connectivity, global observability, and system instability is studied in \cite{khan2010connectivity}. Necessary and sufficient conditions for existence of distributed observers are studied in \cite{park2017design}. Also, \cite{wang2017distributed} generalizes distributed observer design for LTI systems with singular transition matrices. None of the above research works consider communication security among the physical and cyber layers. The resilience of distributed observers against cyber attacks has recently received more attention. For instance, the resiliency of LTI systems has been investigated in \cite{mitra2018resilient,mitra2018secure} against Byzantine attacks. However, our adversarial model is more complex by considering the impersonation capability of adversaries. We also consider asynchrony and delays in communications and propose a randomization strategy for relaxing the imposed topology constraints for secure distributed estimation problem.

In the current paper, we consider impersonation on a network of distributed observers for an LTI system. Like the network communication settings in \cite{dibaji2017automatica,dibaji2017resilient}, the observers communicate with bounded delays and asynchrony; however, they must deal with stronger attacks, i.e. smart spoofing. Similar to \cite{mitra2018resilient}, the regular (un-attacked) nodes are partitioned to source nodes and follower nodes, where source nodes can detect the corresponding eigenvalues and via distributively constructing a directed acyclic graph (DAG), the associated state estimates disseminate through the network. In both DAG construction and estimation propagation, smart spoofers interfere to avoid convergence. We present a strategy based on local filtering that is able to defend against smart spoofing and define local subgraphs to mimic the graph behavior for analysis of the estimation convergence, turning into sufficient conditions on network topology based on graph robustness that is a connectivity measure (see \cite{dibaji2019systems,dibaji2019resilient} for application of similar filtering algorithms in consensus problem). Consistency of the defined spoofing model with the network security literature, consideration of delays, asynchrony, and accurate assumptions in network communications make our proposed algorithms, update rule, and concluding results more practical in real world applications. In the development of our results and the proofs, we adopted the concept of motifs, the smallest possible subgraphs of the original network with certain properties, as a new proof technique. We analyze how the information is disseminated through the motifs. All in all, the main contributions of this paper to the literature are:
\begin{itemize}
\item Introducing, modeling and formulating a new type of cyber attack in asynchronous network settings which inherits the properties of both Byzantine adversaries and spoofing agents, called smart spoofing.
\item Analyzing the vulnerability of asynchronous networks to smart spoofers and proposing a resilient distributed state estimation strategy for a class of LTI systems.
\item Using motifs, as the smallest possible repeating patterns in a network, to mathematically analyze the topology constraints required for convergence of the distributed state estimation.
\item Presenting a randomized update rule to relax the spoof-resilient topology constraint required for convergence of the distributed state estimation.
\end{itemize}

The paper is organized as follows. The preliminaries and problem statement come in section \ref{Sect: Preliminaries}. In section \ref{Sect: ResDisObs}, we take a look at the resilient distributed estimation scheme and the local filtering-based algorithm that we used in this paper. Our main results are presented in Section \ref{Sect: MainResults}. We put forward the simulation results in Section \ref{Sect: Simulations}. Finally, we conclude the paper and discuss the future tendency of the research in Section \ref{Sect: Conclusion}.
 
\section{Preliminaries and Problem Statement}\label{Sect: Preliminaries}

\subsection{Notations}

\subsubsection{Graph Theory}
A directed graph is represented by $\mathcal{G} = (\mathcal{V},\mathcal{E})$, where the set of nodes and edges are represented by $\mathcal{V} = \{ 1, \ldots ,N \}$ and $\mathcal{E} \subseteq \mathcal{V} \times \mathcal{V}$ respectively. An edge from node $j$ pointing to node $i$ implies data transmission from node $j$ to node $i$ and is denoted by $(j,i)$. The neighbourhood of the $i$-th node is defined by the set $\mathcal{N}_i = \{j \vert (j,i) \in \mathcal{E} \}$. A node $j$ is said to be an outgoing neighbour of node $i$ if $(i,j) \in \mathcal{E}$. A spanning sub-graph for $\mathcal{G}$ is a sub-graph of $\mathcal{G}$ which contains every vertex of $\mathcal{G}$. Consider node $v_1$ to $v_p$ of $\mathcal{G}$. A path is a sequence $(v_1,v_2, \ldots ,v_p)$ in which $(v_i,v_{i+1}) \in \mathcal{E}$ for $i=1, \ldots ,p-1$. The length of a path is measured by its number of edges. A cycle is a sequence $(v_1,v_2, \ldots ,v_p,v_1)$ in which $(v_i,v_{i+1}) \in \mathcal{E}$ for $i=1, \ldots ,p-1$ and $(v_p,v_1) \in \mathcal{E}$. A directed acyclic graph (DAG) is a directed graph which has no cycles.

For the consensus-based state estimation rule designed in this paper, the critical topological notion is graph robustness, which is a connectivity measure of graphs (see \cite{leblanc2013resilient}).

\begin{definition} \rm \label{df.reachable}
($r$-reachable set) For a graph $\mathcal{G}=(\mathcal{V},\mathcal{E})$ and a set $\mathcal{C} \subset \mathcal{V}$, we say that $\mathcal{C}$ is an $r$-reachable set if there exists an $i \in \mathcal{C}$ such that $\vert \mathcal{N}_i \setminus \mathcal{C} \vert \geq r$, where $r \in \mathbb{N}_+$.
\end{definition}

\begin{definition} \rm \label{df.robust}
(Strongly $r$-robust w.r.t. $\mathcal{S}$) For a graph $\mathcal{G}=(\mathcal{V},\mathcal{E})$, a set of nodes $\mathcal{S} \subset \mathcal{V}$ and $r \in \mathbb{N}_+$, we say that $\mathcal{G}$ is strongly $r$-robust with respect to $\mathcal{S}$, if for any non-empty subset $\mathcal{C} \subseteq \mathcal{V} \setminus \mathcal{S}$, $\mathcal{C}$ is $r$-reachable. 
\end{definition}

\subsubsection{Linear Algebra}
The set of all eigenvalues of a matrix $A$ is denoted by $\sigma (A)$. The set of all marginally stable and unstable eigenvalues of a matrix $A$ is denoted by $\sigma_U(A)=\{ \lambda \in \sigma (A) \vert \vert \lambda \vert \geq 1 \}$. We use $a_A(\lambda)$ and $
g_A(\lambda)$ to denote the algebraic and geometric multiplicities, respectively, of an eigenvalue $\lambda \in \sigma (A)$. An eigenvalue $\lambda$ is said to be simple if $a_A(\lambda)=g_A(\lambda)=1$.

\begin{figure}[t]
	\def \svgscale{.6}
	\hspace{.5cm}
	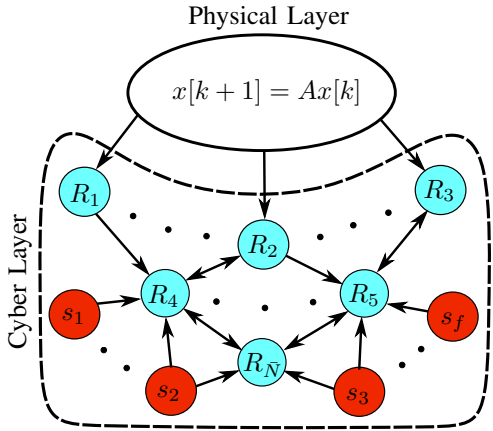
    \caption{A typical cyber-physical system: In the physical layer, the plant dynamic (maybe unstable in some modes) propagates over time. In the cyber layer, the plant's outputs are monitored by a network of distributed observers ($R_1,R_2,\ldots,R_{\bar{N}}$) while only some of them are directly connected to the plant. The network is threatened by adversarial nodes ($s_1,s_2,\ldots,s_f$).} \label{fig.cyphy} 
\end{figure} 

\subsection{System Dynamics and Distributed Observers}
Consider the following discrete-time LTI system

\begin{equation} \label{eq.dyn}
x[k+1]= Ax[k],
\end{equation}
where $k \in \mathbb{N}$ is the discrete-time index, $x[k] \in \mathbb{R}^n$ is the state vector and $A \in \mathbb{
R}^{n\times n}$ is the system matrix. The system is observed by an $N$-node network $\mathcal{G}=(\mathcal{V},\mathcal{E})$. Access of the $i$-th node to the measurement of time instant $k$ is given by

\begin{equation} \label{eq.obs}
y_i[k]=C_i x[k],
\end{equation}
where $y_i[k] \in \mathbb{R}^{r_i}$ and $C_i \in \mathbb{R}^{r_i \times n}$. 
For computational or control purposes, each node needs to estimate the entire system state $x[k]$. Nodes of the network $\mathcal{G}$ are called \textit{distributed observers} if they maintain and update the estimates using only their own measurements and those received from their neighbors. Fig.~\ref{fig.cyphy} shows the layout of a typical cyber-physical system threatened by adversarial nodes. Let $\hat{x}_i[k]$ denote the state estimate of node $i$ at each time step $k$. The following definition describes the objective of the distributed estimation scheme.

\begin{definition} \rm \label{df.omn}
(Omniscience) Over the $N$-node network $\mathcal{G}$, the distributed observers are said to achieve omniscience if $\lim_{k \to \infty} \vert \hat{x}_i[k]-x[k] \vert = 0, \forall i \in \{ 1, 2, \ldots , N \}$.
\end{definition}

\subsection{Adversarial Model}
We consider an adversarial model that is able to threaten the following system protection services: \textit{authentication}, \textit{authorization}, \textit{confidentiality}, \textit{integrity} and \textit{availability}. In what follows, we formally define the abilities of such an adversarial node.

\begin{definition} \rm \label{df.spf}
(Smart Spoofer) An adversarial node is called a \textit{smart spoofer} if it has the following capabilities:

\begin{itemize}
\item[1)] The adversarial node can have complete knowledge about the topology, plant dynamics, and information flow over the network at all time steps.
\item[2)] The adversarial node can refuse to perform any preassigned algorithm and can send arbitrary values to each of its neighbors at the same time step.
\item[3)] The adversarial node can send its data with intended delays and asynchrony.
\item[4)] The adversarial node can impersonate other nodes and send arbitrary data with their identities.
\end{itemize}

\end{definition}

The first two actions are performed by Byzantine adversaries, while the last one is performed by a threat called \textit{spoofing} or \textit{masquerading} in \cite{daswani2007foundations} that directly threatens the authentication among systems' protection services. In fact, the introduced adversarial model is an advanced spoofing threat with additional capabilities of Byzantine adversaries that we call \textit{smart spoofing}. Note that we use the terms \textquotedblleft spoof\textquotedblright \ and \textquotedblleft impersonate\textquotedblright \ interchangeably in this paper.

It is apparent that no distributed estimation algorithm would succeed if all the nodes are adversarial. So, the set of nodes $\mathcal{V}$ is partitioned into two subsets of regular nodes and adversarial nodes denoted by $\mathcal{R}$ and $\mathcal{A} = \mathcal{V} \setminus \mathcal{R}$, respectively. In the literature of distributed fault-tolerant algorithms, a common assumption is to assign an upper bound $f$ to the total number of adversarial nodes in the network, which is known as $f$-total adversarial model. To consider a large number of adversaries in large scale networks, locally bounded fault models are used, as in \cite{pelc2005locallybounded}, defined below.

\begin{definition} \rm \label{df.flmod}
($f$-local smart spoofer model) A set $\mathcal{A}$ of smart spoofers is $f$-locally bounded if it contains at most $f$ smart spoofers in the neighbourhood of any of the regular nodes, i.e. $\vert \mathcal{N}_i \cap \mathcal{A} \vert \leq f, \forall i \in \mathcal{V} \setminus \mathcal{A}$.
\end{definition}

Similarly, any distributed estimation algorithm fails if a smart spoofer can impersonate all the network nodes. Thus, to tackle the problem, we impose an upper bound for the number of nodes that smart spoofers can send data on their behalf as follows.

\begin{definition} \rm \label{df.scap}
(Capacity of smart spoofers) The maximum number of nodes that a smart spoofer can send data on their behalf at each time step, including itself, represents its capacity and is denoted by $\alpha \geq 1$.
\end{definition}

\subsection{Problem Statement}
We aim to formulate the resilient version of omniscience problem (Def.~\ref{df.omn}), where the network is under smart spoofers' attack with two challenging constraints on the network communications, i.e. asynchrony and delays. Accordingly, we set the following assumptions on the network communications protocol remarking the practical aspects of our results.

\begin{assumption} \rm \label{asum1}
All nodes update by a global clock. This means that the sampling time $T$ is the same for all observers.
\end{assumption}

\begin{assumption} \rm \label{asum2}
All nodes communicate through serial links and have access to only the last data packet they have received from neighbor nodes.
\end{assumption} 

\begin{assumption} \rm \label{asum3}
All nodes make, at least, one update within $\bar{k}$ steps and communication delays are upper-bounded by $\bar{\tau}$. 
\end{assumption}

Referring to the introduced LTI dynamic system and the observation model of the network, we put forth a more complicated version of the standard omniscience problem (Definition~\ref{df.omn}) in the following definition.

\begin{definition} \rm \label{df.rsomn}
(Resilient Omniscience) Given a system dynamics of the form (\ref{eq.dyn}), a network represented by the graph $\mathcal{G}$, and an observer model at each node given by (\ref{eq.obs}), a state estimation design is said to achieve \textit{resilient omniscience} if $\lim_{k \to \infty} \vert \hat{x}_i[k]-x[k] \vert = 0, \forall i \in \mathcal{R}$, regardless of the actions of any $f$-locally bounded set of smart spoofers.
\end{definition}

This paper investigates the design of a distributed estimation scheme, proper to cope with smart spoofers threatening a given cyber network that is observing an LTI system. For this purpose, based on the assumptions on the network communications protocol and the smart spoofer adversarial model, we first present the distributed estimation scheme under a specific network topology. Next, we analyze the required topology constraints which guarantee resilient omniscience of all regular nodes that update their estimates using the proposed estimation strategy.

\section{Resilient Distributed Observers}\label{Sect: ResDisObs}
Under Byzantine adversarial model introduced in \cite{lynch1996distributed}, the network achieves omniscience by distributed observers proposed in \cite{mitra2018resilient}. The design performs observation task by separating detectable and undetectable eigenvalues of the system and the related states. Here, we use a similar scheme with a different distributed estimation rule, proper for resilient omniscience defined in Definition~\ref{df.rsomn}. To this end, consider a Jordan canonical decomposition of state transition matrix $A$ with the following assumption on its eigenvalues. This assumption is made for sake of simplicity, is not restrictive, and can be relaxed by some extra mathematical efforts and the techniques denoted in \cite{mitra2018resilient}, which is not the focus of this paper.

\begin{assumption} \rm \label{asum4}
Eigenvalues of $A$ are real and simple.
\end{assumption}
This assumption allows us to diagonalize $A$ by the coordinate transformation matrix $\Psi=[\psi_1, \ldots ,\psi_n]$, where $\psi_1, \ldots ,\psi_n$ are $n$ linearly independent eigenvectors of $A$. With $z[k]=\Psi^{-1}x[k]$, the system (\ref{eq.dyn}) is transformed into the form

\begin{equation} \label{eq.map}
\begin{split}
z[k+1]&=\bar{A}z[k],\\
y_i[k]&=\bar{C}_i z[k], \ \forall i \in \{ 1, \ldots  ,N \},
\end{split}
\end{equation}
where $\bar{A}=\Psi^{-1}A\Psi$ is a diagonal matrix, and $\bar{C}_i=C_i\Psi$. The eigenvalues of $\bar{A}$ (which are the same as those of $A$) are denoted by $\lambda_1, \ldots ,\lambda_n$. Each regular node $i$ distinguishes its detectable and undetectable eigenvalues by PBH test and divides them into the sets $\mathcal{D}_i$ and $\mathcal{U}_i$, respectively. Also, the set of regular nodes are partitioned into sets of \textit{source nodes} and \textit{follower nodes} as defined below.

\begin{definition} \rm \label{df.Sj}
(Source nodes and follower nodes) For each $\lambda_j \in \sigma_U (A)$, the set of nodes that can detect $\lambda_j$ is denoted by $\mathcal{S}_j$, and is called the set of source nodes for $\lambda_j$. The rest of the nodes are called follower nodes.
\end{definition}

Each regular node, depending on being a source node or a follower node for $\lambda_j$, adopts a different strategy for estimating the related states. 

\subsection{State Estimation by Source Nodes}
Referring to \cite{mitra2018resilient}, each regular node $i$ relies on its own measurements and uses a local Luenberger observer to estimate a patch of the states $\hat{z}_{\mathcal{D}_i}$ associated to all $\lambda_j \in \mathcal{D}_i$. To this end, let $\Lambda_i \in \mathbb{R}^{\rho_i \times \rho_i}$ (recall $\rho_i = \vert \mathcal{D}_i \vert$) be a diagonal matrix consists of the detectable eigenvalues in $\mathcal{D}_i$ and $\bar{C}_{\mathcal{D}_i} \in \mathbb{R}^{r_i \times \rho_i}$ stand for the columns of $\bar{C}_i$ corresponding to those eigenvalues. Then we have: 

\begin{equation} \label{eq.lu}
\hat{z}_{\mathcal{D}_i}[k+1]=\Lambda_i \hat{z}_{\mathcal{D}_i}[k] + L_i(y_i[k]-\bar{C}_{\mathcal{D}_i} \hat{z}_{\mathcal{D}_i}[k]),
\end{equation}
where $L_i^j \in \mathbb{R}^{\rho_i \times r_i}$ is the observer gain matrix at node $i$. Since the pair $(\Lambda_i , \bar{C}_{\mathcal{D}_i})$ is detectable, $L_i$ can be chosen in a way that $(\Lambda_i - L_i \bar{C}_{\mathcal{D}_i})$ is Schur stable, so $\lim_{k \to \infty} \vert \hat{z}_i^j[k]-z^j[k] \vert =0$ based on Assumption~\ref{asum4}.

\subsection{Distributed State Estimation by Follower Nodes}
A regular node $i$ cannot estimate a portion of the states associated with its undetectable eigenvalues of the system. In fact, the regular node $i$ is a follower node in estimating the sub-state related to the eigenvalues $\lambda_j \in \mathcal{U}_i$ and needs to receive information from its neighbors through a directed acyclic graph for each $\lambda_j$ (defined later) rooted in the set of associated source nodes. In what follows, we propose an updating rule for the follower node $i$ accomplishing its estimation task in a network with communication delays and \textit{partial asynchrony}\footnote{The term \textit{partial asynchrony} refers to the case where nodes share some level of synchrony by having the same sampling times; however, they make updates at different times based on bounded information delays \cite{bertsekas1989parallel}}. 

There is a major difference between resilient distributed state estimation rather than resilient consensus using local filtering presented in previous research works such as \cite{dibaji2017resilient}. Considering asynchronous network communications and observability of dynamics of the physical layer  is a new challenge in design of the update rule and leads to a totally different convergence analysis. Combining the ideas behind the consensus update rules in \cite{mitra2018resilient} and \cite{dibaji2017resilient}, we present a novel update rule with the following algorithm based on local filtering method for node $i$ to update its own state estimate for $\lambda_j \in \mathcal{U}_i$.

\begin{itemize}
\item[1)] Each regular node $i$, at each time step $k$ when it wants to update its estimate, gathers the state estimate of $z^j[k]$ \textit{lastly} received from only the nodes in $\mathcal{N}_i^j  \subseteq \mathcal{N}_i$ ($\mathcal{N}_i^j$ represents the set of neighbors in the DAG related to $\lambda_j$ that is selected by Algorithm~\ref{alg1} for each regular node $i$, which will be proposed later) and arranges them from the largest to the smallest.
\item[2)] Node $i$ drops the largest and smallest $(\beta +1)f$ estimates ($\beta$ will be defined later) and executes the following update rule:

\begin{equation} \label{eq.consda}
\hat{z}_i^j[k+1]=\lambda_j \sum_{\ell \in \mathcal{N}_i^j} \omega_{i \ell}^j[k] \hat{z}_{\ell}^j[k - \widetilde{k}_{i \ell}[k] - \tau_{i\ell}[k]],
\end{equation}

where $\tau_{i\ell}[k]$ is the time delay of the last data packet that node $i$ has received from node $\ell$ until time instant $k$ (it may be time-varying), $\widetilde{k}_{i \ell}[k]$ is the time steps elapsed from the time that node $i$ receives the packet of the node $\ell$ sent the last time before time $k$ up to the time it makes an update ($\widetilde{k}_{i \ell}[k] < \bar{k}$), and $\omega_{i \ell}^j[k]$ is the weight that the $i$-th node dedicates to the $\ell$-th node at the $k$-th time instant for the estimation of $z^j[k]$. The weights are non-negative and chosen to comply $\sum_{\ell \in \mathcal{N}_i^j} \omega_{i \ell}^j[k]=1, \forall \lambda_j \in \mathcal{U}_i$. Node $i$ removes the $(\beta +1)f$ largest and $(\beta +1)f$ smallest estimates from $\mathcal{N}_i^j$ by setting their associated weights to 0. Note that delays have an upper bound ($\tau_{i \ell} < \bar{\tau}$).
\end{itemize}

In practice, each node $i$ has a memory for each of its neighbors where stores the most recently received data. Node $i$ uses the most recent estimate values received from its neigbours in $\mathcal{N}_i^j$ in update rule (\ref{eq.consda}), regardless of delays and asynchrony in communications.

\section{Main Results}\label{Sect: MainResults}
In this section, we provide the main results of the paper giving the analysis of the spoof-resilient distributed estimation strategy and the topology constraints under which our adopted algorithms and update rule succeed. 

First, we consider how harsh the misbehaviour effects of a smart spoofer would be in the network. In Definition~\ref{df.scap}, we introduced the spoofing capacity in each time step. In the following lemma, we generalize capacity of smart spoofers for a period of time.

\begin{lemma} \rm \label{lm.kscap}
Let Assumption~\ref{asum3} hold and capacity of a smart spoofer be $\alpha$. Then, each smart spoofer is able to send data on behalf of $\beta = \alpha \bar{k} -1$ regular nodes within each consecutive $\bar{k}$ steps.
\end{lemma}

\begin{proof}
According to Definition~\ref{df.scap}, a smart spoofer can send data on behalf of $\alpha$ nodes including itself at each time step. Also, according to Assumption~\ref{asum3}, all nodes have to make at least one update within consecutive $\bar{k}$ steps (note that if a node does not follow this rule can be detected as an adversarial node by the regular nodes). Consider the time interval $k+1 \leq t \leq k+\bar{k}$. Let the smart spoofer choose to make an update with its own identity at $t=k_s$, where $k+1 \leq k_s \leq k + \bar{k}$. Considering each consecutive $\bar{k}$ steps, the smart spoofer $s$ has $\alpha -1$ capacity for impersonation at $k_s$ and $\alpha$ capacity in other $\bar{k}-1$ steps. Therefore, the smart spoofer is able to impersonate $\alpha \bar{k}-1$ nodes overall within $\bar{k}$ steps, i.e. $\alpha - 1 + \alpha (\bar{k}-1)=\alpha \bar{k}-1$.
\end{proof}

In fact, Lemma~\ref{lm.kscap} indicates that we cannot simply replace the smart spoofers and impersonated nodes with Byzantine nodes. Because the key question is that how many Byzantine nodes have the same effect of a smart spoofer with capacity $\alpha$. This is what we mathematically clarified in Lemma 1. Asynchrony lets the adversarial nodes spoof a specific number of regular nodes within each $\bar{k}$ time-steps. Besides, from adversarial nodes' perspective, this spoofing (and sending false data packets) must be continued for all the future time -- in every $\bar{k}$ steps – in order to be effective. Thus, the distributed algorithms of regular nodes must be modified to be resilient against the attack. In other words, network providers need to be aware that in a network with asynchrony (almost all the networks are practically asynchronous), there is the possibility of stronger attacks rather than Byzantine attacks. The following necessary condition on the network communications  formally states when a smart spoofer can impersonate regular nodes.

\begin{proposition} \rm \label{prop.nc}
Consider a network of nodes interconnected by complete graph $\mathcal{G}$, which contains smart spoofer $s \in \mathcal{N}_i^j$, where $i \in \mathcal{R}$. Suppose that $s$ is able to impersonate, at least, one regular node within each consecutive $\bar{k}$ steps ($\beta >0$). Smart spoofer $s$ can impersonate a regular node $\ell \in \mathcal{N}_i^j \cap \mathcal{R}$ for node $i$ at a time instant $t>k$ only if the packet which is broadcast by node $\ell$ at time instant $t=k$ is received by node $i$ with delay $\widetilde{k}_{i \ell}[k]+\tau_{i\ell}[k]>0$.
\end{proposition}

\begin{proof}
We prove by contradiction. Considering Assumption~\ref{asum1}, let node $\ell \in \mathcal{N}_i^j$ broadcast a data packet at $t=k$ and node $i \in \mathcal{R}$ receives the packet at the same time ($\tau_{i\ell}[k]=0$) and use it for its next update at the same time ($\widetilde{k}_{i \ell}[k]=0$). Also, suppose that smart spoofer $s$ decides to impersonate node $\ell$ for node $i$. There are two possibilities for the arrival time of the packet sent by $s$ to $i$. The packet can arrive either before or after the time $t=k$ (the time instant $t=k$ is excluded as it contradicts the Assumption~\ref{asum2}). In case the packet sent by $s$ arrives at any time $t>k$, the node $i$ has already accepted the last packet it received, that is, the packet of node $\ell$ received at $t=k$, and has already made an update. Otherwise, if the packet sent by $s$ arrives at any time $t<k$, then $i$ will receive the packet sent by node $\ell$ at $t=k$ and since, according to Assumption~\ref{asum2}, all nodes only access the last data packet they receive. In either case, the smart spoofer $s$ fails to impersonate node $\ell$ for node $i$, which is a contradiction. This completes the proof.
\end{proof}

Note that the necessary condition of Proposition~\ref{prop.nc} is independent of amounts of communication delays. This is because we would like to deal with smart spoofers that can impose arbitrarily bounded amount of delays on their links to the regular nodes, that is, if a smart spoofer wants to impersonate node $\ell \in \mathcal{N}_i^j$, it can arrange to send the packet to node $i$ after node $\ell$ with appropriate delay so that it will be received after the packet sent by node $\ell$. Then, node $i$ accepts a packet sent by $s$ with identity of node $\ell$ as it is the last packet received.

According to Proposition~\ref{prop.nc}, the best case for the regular node $i$ is that both $\widetilde{k}_{i \ell}[k]=0$ and $\tau_{i\ell}[k]=0$, so the smart spoofers in the neighborhood of $i$ cannot impersonate neighbors of the node $i$. However, even if we suppose that $\widetilde{k}_{i \ell}[k]=0$, i.e. node $i$ has not any lag in updating its estimate using the last data received from node $\ell$, regular nodes cannot be sure about spoofing attack. In practice, the regular nodes cannot guess, before receiving a packet, whether it will be received with delay and, if so, how much the delay will be (although communication links' delays in real network systems are inevitable). Besides, as we said, smart spoofers can send data packets with intended delays. So, the regular nodes must be aware that all the communications may be done with delays in each time step. Therefore, to consider the worst case, we develop our further results on required topology constraints by assuming that the necessary condition on delays is satisfied for all time in the network.

\subsection{Spoof-Resilient Mode Estimation Directed Acyclic Graph (SR-MEDAG)}
Recall the local filtering for resilient consensus based estimation law (\ref{eq.consda}). Inspired by the algorithm presented in \cite{mitra2018resilient}, for construction of directed acyclic graphs associated with undetectable eigenvalues of an LTI system, we present a spoof-resilient algorithm which is distributively executed by all the regular nodes. The overall distributed estimation scheme constitutes the construction of these sub-graphs and the prescribed local filtering-based algorithm which are performed in parallel. In what follows, we define the directed acyclic graphs that are paths for information flow over the network.
\begin{algorithm}[t]
\caption{SR-MEDAG Construction Algorithm}\label{alg1}
\begin{algorithmic}[1]
 \While{$k \leq \bar{K}_j$}
   \For{$\lambda_j \in \sigma_U (A)$}
  	$c_i(j)=0, \mathcal{N}_i^j=\emptyset$.\\
  	\If{$i \in \mathcal{S}_j$}
   	\State Node $i$ updates $c_i(j)$ to 1, sets $\mathcal{N}_i^j=\emptyset$ and broadcasts a message
   	\State  $\chi_j$ (e.g. ``110") to its neighbors.
  	\EndIf
  	\If{$i \in \mathcal{V} \setminus \mathcal{S}_j$}
  	\If{($c_i(j)=0$, and node $i$ has received $m_j$ from at least $2(\beta +1)f +1$ distinct neighbors)}
   	\State Node $i$ updates $c_i(j)$ to 1 and stores the labels of the 
   	\State neighbors from which it received $\chi_j$ to $\mathcal{N}_i^j$.
  	\EndIf
  	
  	\If{$c_i(j)=1$}  	
   	\State Node $i$ broadcasts $\chi_j$ to its neighbors.
  	\EndIf
  	
  	\EndIf
 \EndFor
 \EndWhile
 \State {\bf Result}: $ \mathcal{N}_i^j$, $\forall \lambda_j \in \sigma_U (A)$
\end{algorithmic}
\end{algorithm} 

\begin{definition} \rm \label{df.SR-MEDAG}
(SR-MEDAG) For each eigenvalue $\lambda_j \in \sigma_U (A)$, the spanning sub-graph $\mathcal{G}^j = (\mathcal{V},\mathcal{E}^j)$ of $\mathcal{G}$ is Spoof-Resilient Mode Estimation Directed Acyclic Graph (SR-MEDAG) if it has the following properties:
\begin{itemize}
\item[1)] If $i \in (\mathcal{V} \setminus \mathcal{S}_j) \cap \mathcal{R}$, then $\vert \mathcal{N}_i^j \vert \geq 2(\beta +1) f+1$.
\item[2)] There is a partition of $\mathcal{R}$ into the sets $\mathcal{L}_0^j, \mathcal{L}_1^j, \ldots , \mathcal{L}_{\xi_j}^j$, such that $\mathcal{L}_0^j = \mathcal{S}_j \cap \mathcal{R}$, and $ \mathcal{N}_i^j \cap \mathcal{R} \subseteq \bigcup_{r=0}^{m-1} \mathcal{L}_r^j$ for $i \in \mathcal{L}_m^j$, where $1 \leq m \leq \xi_j$.
\end{itemize}
Also, $ \mathcal{N}_i^j \cap \mathcal{R}$ is the set of parent nodes of node $i$ and $\mathcal{L}_m^j$ is the $m$-th layer of $\mathcal{G}^j$. In fact, for each $\lambda_j$, we can organize the set of regular nodes of the graph $\mathcal{G}$ as a directed acyclic graph $\mathcal{G}_j$. In $\mathcal{G}_j$, the set of regular source nodes are denoted by $\mathcal{L}_0^j$. Also, the set of regular nodes which have at least one path with length of $m$ to $\mathcal{L}_0$ are in the $m$-th layer of $\mathcal{G}_j$ denoted by $\mathcal{L}_m$. Each regular node in $m$-th layer has at least $2(\beta +1)f+1$ parent nodes from the previous layers ($\bigcup_{r=0}^{m-1} \mathcal{L}_r^j$).
\end{definition}

Each regular node $i \in \mathcal{R}$ distributively executes the SR-MEDAG construction algorithm, presented as Algorithm~\ref{alg1}. The final result of the algorithm for node $i$ is the set $\mathcal{N}_i^j$ associated with every undetectable eigenvalue $ \lambda_j \in \sigma (A)$. By executing the algorithm at each time step $k$, node $i$ stores a counter value $c_i(j)$ and a list of indices $\mathcal{N}_i^j$ in persistent memories for each undetectable eigenvalue $\lambda_j$. The stored values in $\mathcal{N}_i^j \subseteq \mathcal{N}_i$ are the parent nodes' indices of node $i$ in the SR-MEDAG of $\lambda_j$. Each regular node $i$ starts with $c_i(j)=0$ and $\mathcal{N}_i^j=\emptyset$. If node $i$ was a source node for $\lambda_j$, i.e. $i \in \mathcal{S}_j$, it sets $c_i(j)=1$ and $\mathcal{N}_i^j=\emptyset$, then it begins and keeps broadcasting an arbitrary preset message $\chi_j$ to its neighbors\footnote{We used the term \textit{broadcast} considering the case of wireless networks. Regular nodes may transmit the information to their known outgoing neighbors in wired networks.} for at least $\bar{K}_j$ steps (later we prove that $\bar{K}_j$ is bounded). If node $i$ was a follower node for $\lambda_j$, i.e. $i \in \mathcal{V} \setminus \mathcal{S}_j$, it waits until it receives $\chi_j$ from at least $2(\beta +1)f+1$ distinct neighbors. Then, it sets $c_i(j)=1$, saves the indices of the neighbors from which it received $\chi_j$ as $\mathcal{N}_i^j$, begins and keeps broadcasting $\chi_j$ to its neighbors for at least $\bar{K}_j$ steps. Finally, we say that \textit{SR-MEDAG construction algorithm terminates for node $i$} if the counter value $c_i(j)=1$, $\forall \lambda_j \in \sigma_U (A)$. Also, we say that \textit{SR-MEDAG construction phase terminates for $\lambda_j$} if the counter value $c_i(j)=1$, $\forall i \in \mathcal{R}$.

Interestingly, it is not necessary for the regular nodes to know $\bar{K}_j$ (in that case, they have to execute the construction algorithm for all the future time not up to $\bar{K}_j$). Indeed, each regular node $i$ can begin updating its state estimates in parallel as soon as it sets $c_i(j)=1$ for $\lambda_j$ although the SR-MEDAG construction phase has not been terminated yet. However, we know that the construction phase will be terminated at some time instant in the future (bounded by $\bar{K}_j$) when all regular nodes will be able to update their own state estimates corresponding to each of the undetectable eigenvalues using the distributed consensus-based rule (\ref{eq.consda}). In this regard, consider that delay and asynchrony do not affect the output of the algorithm for each regular node $i$. Because node $i$ waits until it receives the predefined message $\chi_j$ from a specified number of nodes regardless of the time it takes. Indeed, asynchrony and bounded delays only postpone termination of the algorithm. 

Furthermore, one may concern that some of the regular nodes are exposed to be impersonated by smart spoofers at any time while they are executing the construction algorithm. In fact, each smart spoofer not only can impersonate regular nodes (by sending arbitrary messages other than the true message $\chi_j$ on behalf of them) but also can misbehave as follows:
\begin{itemize}
\item[i)] It chooses to transmit any message different from the true message $\chi_j$ from start to termination of the construction phase.
\item[ii)] It transmits the true message before the counter value is triggered by a regular node.
\item[iii)] It chooses not to transmit a message at all.
\end{itemize}
In the first case, regular nodes are able to identify the adversarial node as it goes against the rules of Algorithm~\ref{alg1}. In the latter two cases, the adversarial node is undetectable by regular nodes relying just on local information. However, later we discuss constraints on the graph topology so that adversarial nodes fail to make any problem neither for the construction algorithm nor the estimation process.

It is noteworthy that the upper bound for the parameter $\bar{K}_j$ in Algorithm~\ref{alg1} is a function of the parameter $\beta$, which is the capacity that asynchrony provides for smart spoofers to impersonate regular nodes. This upper bound would be different if we consider simply more Byzantine nodes instead of spoofers and impersonated nodes. Actually, another contribution of our paper with respect to \cite{mitra2018resilient} is the MEDAG construction algorithm and its convergence time. In the case of synchronous networks, each regular node updates only once and goes to sleep, while, in asynchronous networks, regular nodes have to continue updating up to $\bar{K}_j$ time-steps in order to complete the SR-MEDAG construction.
     
In the following theorem, we show that the sub-graphs distributively found by the regular nodes, after termination of the construction phase, satisfy properties of the SR-MEDAG.

\begin{theorem} \rm \label{th.SRMEDAG}
If the SR-MEDAG construction phase terminates for $\lambda_j \in \sigma_U(A)$, there exists a sub-graph $\mathcal{G}^j$ satisfying all the properties of an SR-MEDAG.
\end{theorem}

\begin{proof}
First, we prove by contradiction that the spanning sub-graph $\mathcal{G}^j$ is a directed acyclic graph. Suppose there is a directed cycle $iPi$, where $i$ and the nodes in $P$ belongs to $\mathcal{R}$. The path $P$ originates from $i$ which changes its counter value $c_i(j)$ from 0 to 1 and begins transmitting $\chi_j$ to its neighbors at a time instant $t=k_i^j$. Let the last node on the path $P$ be $\ell$. Clearly, node $i$ receives data from node $\ell$ at a time instant $t>k_i^j$. As an edge from $\ell$ is pointing to node $i$, node $i$ is supposed to receive the message $\chi_j$ from node $\ell$ even when its counter value $c_i(j)$ is set to $1$. This contradicts what node $i$ has to do according to Algorithm \ref{alg1}. The same argument holds for every regular node belonging to $\mathcal{G}^j$. 

Next, we associate the notion of path length, referring to graph theory, to the found sub-graphs after the termination of the SR-MEDAG construction phase to show that the set $\mathcal{R}$ in $\mathcal{G}^j$ is partitioned to the sets $ \mathcal{L}_0^j, \mathcal{L}_1^j, \ldots , \mathcal{L}_{\xi_j}^j $. To this end, let a regular node $i$ update its counter value $c_i(j)$ from 0 to 1 at a time instant $t=k_i^j$. Then, we say that the node $i$ belongs to $\mathcal{L}_m^j$ of $\mathcal{G}^j$ if length of its longest path to a node in $S_j$ be $m$ at $k_i^j$. Apparently, $\mathcal{L}_{\xi_j}^j$ is set of the nodes which have at least a path with maximum length (among all acyclic paths of $\mathcal{G}^j$) to a node in $\mathcal{S}_j$ as $1 \leq m \leq \xi_j$. Accordingly, node $i$ belongs to $\mathcal{L}_0^j$ of $\mathcal{G}^j$ if $i \in \mathcal{S}_j \cap \mathcal{R}$. Now, suppose that the SR-MEDAG construction phase terminates for $\lambda_j \in \sigma_U (A)$. Since all the nodes update their counter values from 0 to 1 at some time instant, it is concluded that $\bigcup_{r=0}^{\xi_j} \mathcal{L}_r^j = \mathcal{R}$. Moreover, a regular node in $\mathcal{R}$ cannot update its counter value at two different time steps (the converse contradicts the rules of Algorithm~\ref{alg1}). Thus $\mathcal{L}_r^j \cap \mathcal{L}_s^j = \emptyset, \ \forall r \neq s$. This completes the proof according to the definition of the sets $\mathcal{L}_m^j \ (0 \leq m \leq \xi_j)$.
\end{proof}

\begin{remark} \rm
Since the network communications are asynchronous and because each regular node does not know the communication delays between other nodes, regular nodes in the sets $\mathcal{L}_0^j, \mathcal{L}_1^j, \ldots ,\mathcal{L}_{\xi_j}^j$ do not update their counter values in the same order as their layer number.
\end{remark}

We intentionally used the minimum number of variables to be communicated in SR-MEDAG so as to avoid potential masquerading threats caused by those variables. For example, it is not possible for regular nodes to realize their layer order in $\mathcal{G}^j$ as they cannot identify which of their parent nodes are spoofed. To clarify this, consider a regular node $i$ in $\mathcal{L}_m^j$. The regular node has to receive $2(\beta+1)f+1$ incoming edges from the nodes in $\bigcup_{r=0}^{m-1} \mathcal{L}_r^j$ which broadcast their layer numbers so that the node $i$ can realize its own layer number by sorting the received values and selecting the maximum as the previous layer number. However, smart spoofers can impersonate some of these nodes and send a wrong layer number behind of them. Thus, the node $i$ can be deceived about the maximum layer number it received. Interestingly, in our method, there is no need that the regular nodes know their layer orders since they only needs to know $2(\beta+1)f+1$ of their parent nodes to succeed in the estimation phase. Therefore, the construction algorithm can still be executed distributedly. Moreover, our strategy succeeds even if some of the source nodes in $\mathcal{S}_j$ are smart spoofers.

\begin{figure}  
	\def \svgscale{.6}
	\hspace{.5cm}
	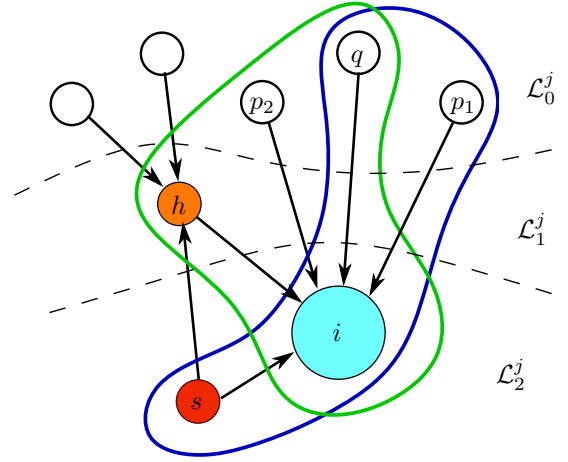
    \caption{Motifs found in SR-MEDAG of $\mathcal{G}^j$ for $f=1$ and $\beta =1$ with the regular node $i$ in $\mathcal{L}_2^j$, three regular parent nodes of $i$ in $\mathcal{L}_0^j$ ($p_1$ and $p_2$ are independent nodes of the motifs and $q$ is a common node) and a parent node in $\mathcal{L}_1^j$ which can be impersonated by a smart spoofer. Here, node $h$ is impersonated by the smart spoofer $s$ for the node $i$.} \label{fig.motif} 
\end{figure} 

\subsection{Analysis of the Resilient Distributed Estimation Strategy}

In this subsection, we introduce a repeating pattern sub-graph which is used to simplify the analysis of our distributed estimation scheme. These sub-graphs are constructed and organized for each regular node according to its incoming edges from smart spoofers and other regular neighbors. Note that they may have overlaps in specific nodes and are defined as follows.  

\begin{definition} \rm \label{df.motif}
(Motifs) Consider a regular node $i \in \mathcal{L}_m^j$ at time instant $k$ and $\lambda_j \in \sigma_U(A)$. Partition $\mathcal{N}_i^j$ into subsets $\{ q \}$, $\{ p_r \}$, $r=\{1,2,\ldots, \underline{r} \}$, and $\{ h_l \}$, $l=\{1,2,\ldots, \bar{l} \}$, where $q$ and $p_r$ are the parent nodes of node $i$ that are not impersonated and $h_l$ is a smart spoofer in $\mathcal{N}_i^j$ or an impersonated parent node of   the node $i$. Then, $\mathcal{G}_{i}^j(l,r) = \big(\mathcal{V}_{i}^j(l,r),\mathcal{E}_i^j(l,r) \big)$ is a sub-graph of $ \mathcal{G}^j$ indicating the motif associated with $h_l$ and $p_r$ around the node $i$, where $\mathcal{V}_{i}^j(l,r) = \{ i,p_r,q,h_l \}$ and $\mathcal{E}_i^j(l,r) = \{ (p_r,i),(q,i),(h_l,i) \}$.
\end{definition}

We aim to associate each motif $\mathcal{G}_{i}^j$ to each node $i$ to ensure that a smart spoofer or an impersonated node cannot deviate the estimation of the node $i$. In fact, each motif is the smallest sub-graph of $\mathcal{G}^j$ which is resilient against Byzantine attacks. Note that impersonated regular nodes are potential Byzantine adversaries since smart spoofers can use their identities to send arbitrary values to their neighbors.

\begin{definition} \rm \label{df.incom}
(Independent and common nodes) Consider $\bar{\gamma}$ motifs associated with $h_l$ and $p_r$ around the regular node $i$ denoted by $\mathcal{G}_{i}^j(l,r)$, $r = 1,2, \ldots ,\bar{\gamma}$. Let $p_{\gamma}$ be a regular node in the set $\mathcal{V}_j^i(l,r)$, $r=\gamma$. Then, $p_{\gamma}$ is an independent node if $p_{\gamma} \notin \mathcal{V}_j^i(l,r) \ \forall r \neq \gamma$. A node that is not independent is called common.
\end{definition}

The analysis strategy is to find the set of motifs around node $i \in \mathcal{L}_m^j$ such that they have only one common node. The following lemma determines the number of such motifs and investigates the possibility of this strategy (see Fig.~\ref{fig.motif} for an example).

\begin{lemma} \rm \label{lm.motif}
Consider the network $\mathcal{G}$ which contains an SR-MEDAG $\mathcal{G}^j$ for each $\lambda_j \in \sigma_U (A)$. There exist at least $(\beta +1) f$ motifs around each regular node $i \in \mathcal{L}_m^j$, where each motif has at least an independent node. 
\end{lemma}

\begin{proof}
For each regular node $i \in \mathcal{L}_m^j$, referring to Definition~\ref{df.motif}, consider partitioning of $\mathcal{N}_i^j$ into subsets $\{ q \}$, $\{ p_r \}$, $r=\{1,2,\ldots, \underline{r} \}$, and $\{ h_l \}$, $l=\{1,2,\ldots, \bar{l} \}$. Based on the first property of SR-MEDAG $\mathcal{G}^j$, we have $ \mathcal{N}_i^j \geq 2(\beta +1)f+1$. Under $f$-local smart spoofer model, there are at most $f$ smart spoofers around the node $i$, i.e. $ \vert \mathcal{N}_i^j \cap \mathcal{R} \vert \geq (2 \beta +1) f +1$. These regular nodes are parent nodes of $i$ based on the second property of SR-MEDAG $\mathcal{G}^j$. According to Lemma~\ref{lm.kscap}, at most $\beta f$ of these parent nodes may be impersonated by the smart spoofers. Thus, at most $(\beta +1)f$ of the nodes in $\mathcal{N}_i^j$ are whether smart spoofers or impersonated parent nodes of the node $i$ which are partitioned as $\{ h_l \}$, $l=\{1,2,\ldots, \bar{l} \}$, i.e. $\bar{l} \leq (\beta +1)f$. We can organize at least $(\beta +1)f+1$ of the remaining parent nodes, which cannot be impersonated, to construct the motifs around the node $i$. Based on Definition~\ref{df.motif} and Definition~\ref{df.incom}, we pick a parent node $q$ as a common node and leave the rest in the set of independent parent nodes $\{p_r\}$, $r=1,2, \ldots , \underline{r}$, i.e. $\underline{r} \geq (\beta +1)f$. Since $\underline{r} \geq \bar{l}$, we can find at least $(\beta +1) f$ motifs around the regular node $i$ such that all of them have one common parent node $q$ and each associated with an independent parent node $p_r$ and a node $h_l$. This completes the proof.
\end{proof}

Fig.~\ref{fig.motif} exhibits two overlapping motifs. In this example, there is a smart spoofer node around the regular node $i$, i.e. $f=1$. Also, it is supposed that $\beta=1$. Thus, according to Lemma~\ref{lm.motif}, two motifs are found around node $i$. Note that smart spoofer $s$ can impersonate node $h$, so a motif has to be constructed with the node $h$ as the adversarial node. Also, node $q$ is selected as the common node while $p_1$ and $p_2$ are the two independent parent nodes of node $i$.

Next, we use the notion of motifs to analyze estimation resilience of the network $\mathcal{G}$ enhanced by the distributed estimation update rule (\ref{eq.consda}) at each node against the adversarial nodes in the presence of communication delays and asynchrony. We start with regular node $i \in \mathcal{L}_m^j$ and generalize our analysis to the whole network afterwards.

\begin{lemma} \rm \label{lm.sandwich}
Consider the network $\mathcal{G}$ which contains an SR-MEDAG $\mathcal{G}^j$ for each $\lambda_j \in \sigma_U(A)$. Suppose that the state estimates of regular parent nodes of node $i \in \mathcal{L}_m^j$, for the state related to $\lambda_j$, converge to $z^j$ asymptotically. Then, the local filtering-based algorithm governed by update rule (\ref{eq.consda}) ensures that $\lim_{k \to \infty} \vert \hat{z}_i^j[k]-z^j[k] \vert =0$ in the presence of communication delays and asynchrony under $f$-local smart spoofer model.
\end{lemma}

\begin{proof}
Based on Lemma~\ref{lm.motif}, node $i$ has at most $(\beta +1)f$ potential threats and at least $(\beta +1)f$ motifs around node $i$ such that they have one common node. Consider the motifs around node $i$: $\mathcal{G}_{i}^j(l,r)$, $r = 1,2, \ldots ,\underline{r}$, $l = 1,2, \ldots ,\bar{l}$, where $\underline{r} \geq (\beta +1)f$ and $\bar{l} \leq (\beta +1)f$. Let $p_r$ be an independent parent node of the node $i$, $q$ the common parent node, and $h_l$ the potential Byzantine adversarial node in $\mathcal{G}_{i}^j(l,r)$, where the state estimation of $p_r$ and $q$ converge to $z^j$ asymptotically for the state related to $\lambda_j$. For simplicity of notations, we define $k_{i\ell} = k - \widetilde{k}_{i \ell}[k] - \tau_{i\ell}[k]$. Note that $\tau_{i\ell}[k]$ and $\widetilde{k}_{i\ell}$ are upper bounded by $\bar{\tau}$ and $\bar{k}$, thus $k_{i\ell} = k$ if $k \to \infty$. Therefore, for each node $i$, we calculate the \textit{asymptotic estimation error} of node $\ell \in \mathcal{N}_i^j$ for $z^j[k]$ by the last transmitted data to the node $i$: $e_\ell^j[k] = \hat{z}_\ell^j[k_{i\ell}]-z^j[k]$, $k \to \infty$. Also, the asymptotic estimation error of the node $i$ for $z^j[k]$ is denoted by $e_i^j[k] = \hat{z}_i^j[k]-z^j[k]$, $k \to \infty$. Then, subtracting $z^j[k+1]$ from both sides of (\ref{eq.consda}) and noting that $z^j[k+1]=\lambda_j z^j[k]$ based on (\ref{eq.map}), we derive Equation~\ref{eq.sandwich} from (\ref{eq.consda}). Equation (\ref{eq.sandwich}) represents that the estimation error of the node $i$ is a linear combination of the estimation errors of its neighbors which are grouped as motifs. Note that $\sum_{\ell \in \mathcal{N}_i^j} \omega_{i\ell}^j[k]=1$. For the un-impersonated parent nodes of node $i$, we have $\lim_{k \to \infty} e_q^j[k]=0$ and $\lim_{k \to \infty} e_{p_r}^j[k]=0$, $r=1,2,\ldots , \underline{r}$. 

Now, consider the motifs $\mathcal{G}_{i}^j(l,r)$, $l=r=1,2,\ldots , \bar{l}$ (note that $\underline{r} \geq \bar{l}$) for the adversarial nodes $h_\ell$, $l=1,2,\ldots , \bar{l}$. In construction of the motifs, we arbitrarily pick the common node $q$; so we suppose that $\hat{z}_{p_r}^j[k] \leq \hat{z}_q^j[k]$, $r=1,2,\ldots , \bar{l}$. The following two cases are possible regarding the estimation values of the nodes in the motif $\mathcal{G}_{i}^j(l,r)$: (i) $\hat{z}_{h_l}^j[k] < \hat{z}_{p_r}^j[k_{ip_r}]$ or $\hat{z}_{h_l}^j[k] > \hat{z}_q^j[k]$, (ii) $\hat{z}_{p_r}^j[k] \leq \hat{z}_{h_l}^j[k] \leq \hat{z}_q^j[k]$. In the former case, according to the local filtering algorithm, values of the node $h_l$ will be removed by setting $\omega_{ih_l}[k]=0$. From the latter case, we infer that $e_{p_r}^j[k] \leq e_{h_l}^j[k] \leq e_q^j[k]$; the asymptotic estimation error of the adversarial node $h_l$ will be trapped by the estimation errors of the parent nodes $p_r$ and $q$ in motif $\mathcal{G}_{i}^j(l,r)$ at time step $k$ and will be sandwiched by them over time as they converge to $0$ asymptotically. Therefor, we conclude that $\lim_{k \to \infty} e_{h_l}^j[k]= 0$. 

\begin{equation} \label{eq.sandwich}
\begin{aligned}
e_i^j[k+ & 1]=\\
& \lambda_j \sum_{\ell \in \mathcal{N}_i^j} \omega_{i\ell}^j[k]\hat{z}_\ell^j[k_{i\ell}] - \lambda_j \bigg( \sum_{\ell \in \mathcal{N}_i^j} \omega_{i\ell}^j[k] \bigg) z^j[k] \\
=&\ \sum_{l=r=1}^{\bar{l}} \bigg( \lambda_j \omega_{ip_r}^j[k]\hat{z}_{p_r}^j[k_{i p_r}] - \lambda_j \omega_{ip_r}^j[k] z^j[k] \\
& + \lambda_j \omega_{ih_l}^j[k]\hat{z}_{h_l}^j[k_{i h_l}] - \lambda_j \omega_{ih_l}^j[k] z^j[k] \bigg) \\
&+ \sum_{r=\bar{l}+1}^{\underline{r}} \bigg( \lambda_j \omega_{ip_r}^j[k]\hat{z}_{p_r}^j[k_{i p_r}] - \lambda_j \omega_{ip_r}^j[k] z^j[k] \bigg) \\
& + \lambda_j \omega_{iq}^j[k]\hat{z}_q^j[k_{iq}] - \lambda_j \omega_{iq}^j[k] z^j[k]\\
=&\ \lambda_j \sum_{l=r=1}^{\bar{l}} \bigg( \omega_{ip_r}^j[k]e_{p_r}^j[k] + \omega_{ih_l}^j[k]e_{h_l}^j[k] \bigg) \\
& + \lambda_j \sum_{r=\bar{l}+1}^{\underline{r}} \omega_{ip_r}^j[k]e_{p_r}^j[k] + \lambda_j \omega_{iq}^j[k]e_q^j[k].
\end{aligned}
\end{equation}

The same argument holds for all adversarial nodes $h_l \in \mathcal{G}_{i}^j(l,r)$, $l=1,2,\ldots , \bar{l}$. Therefore, the estimation error $e_i^j[k+1]$, which is the linear combination of the estimation errors $e_{p_r}^j[k]$, $r=1,2,\ldots , \underline{r}$,  $e_{h_l}^j[k]$, $l=1,2,\ldots , \bar{l}$ and $e_q^j[k]$, converges to $0$ asymptotically, i.e. $\lim_{k \to \infty} e_i^j[k] = \lim_{k \to \infty} \vert \hat{z}_i^j[k]-z^j[k] \vert =0, \forall i \in \mathcal{L}_m^j$. 
\end{proof}

Now, we analyze resilience of the estimation of all the follower nodes in the whole network $\mathcal{G}$ with communication delays and asynchrony against $f$-local smart spoofer model.

\begin{lemma} \rm \label{lm.netest}
Consider the network $\mathcal{G}$ which contains an SR-MEDAG $\mathcal{G}^j$ for each $\lambda_j \in \sigma_U (A)$. Then, for each regular node $i \in \mathcal{R}$ and each $\lambda_j \in \mathcal{U}_i$, the local filtering-based algorithm governed by update rule (\ref{eq.consda}) ensures that $\lim_{k \to \infty} \vert \hat{z}_i^j[k]-z^j[k] \vert =0$ in the presence of communication delays and asynchrony under $f$-local smart spoofer model.
\end{lemma}

\begin{proof}
As $\mathcal{G}$ contains an SR-MEDAG for each $\lambda_j \in \sigma_U (A)$, the sets $ \mathcal{L}_0^j, \mathcal{L}_1^j, \ldots  , \mathcal{L}_{\xi_j}^j $ form a partition of the set $\mathcal{R}$. To prove, we use induction on the layer number $m$. 

For $m=0$, by definition of the set $\mathcal{L}_0^j$, all the regular nodes in $\mathcal{L}_0^j$ belong to the set $\mathcal{S}_j$ and can estimate $z^j[k]$ asymptotically. Then, consider the regular node $i$ belonging to the set $\mathcal{L}_1^j$. Suppose that the regular node $i$ has $f$ incoming edges from adversarial nodes. Then, according to Lemma \ref{lm.motif}, we can find at least $(\beta +1)f$ motifs around node $i$ with each of them having an independent regular source node from the set $\mathcal{L}_0^j$. Each smart spoofer is able to impersonate at most $\beta f$ parent nodes of $i$. So, there are at most $(\beta +1)f$ Byzantine adversarial nodes (including smart spoofers and impersonated regular nodes) around node $i$. However, we infer from Lemma \ref{lm.sandwich} that the state estimate value of each adversarial node is trapped and sandwiched by one of the motifs according to the local filtering algorithm. Thus, the state estimate of node $i$ converges to $z^j$ asymptotically.  

Next, suppose the result holds for the regular nodes of all layers from $0$ to $m$ (where $1 \leq m \leq \xi_j - 1$). By induction, it is concluded that the result holds for all the regular nodes in $\mathcal{L}_{m+1}^j$ as well based on the definition of SR-MEDAG.
\end{proof}

Due to the linear dynamics of the local Luenberger observers for source nodes and since the estimation error of each follower node is a linear combination of its un-impersonated parents, we infer the following corollary about the convergence rate of the follower nodes' estimation error.
\begin{corollary} \rm
Estimation convergence rate of all the follower nodes in the network is exponential as the estimation error of the source nodes converges to $0$ exponentially.
\end{corollary}

\begin{theorem} \rm \label{th.est}
Consider the network $\mathcal{G}$ which contains an SR-MEDAG for each $\lambda_j \in \sigma_U(A)$. Then, the distributed estimation scheme governed by the Luenberger observers described by (\ref{eq.lu}), and the local filtering-based algorithm governed by update rule (\ref{eq.consda}), achieves resilient omniscience in the presence of communication delays and asynchrony under $f$-local smart spoofer model.
\end{theorem}

\begin{proof}
Based on the observable canonical decomposition represented by (\ref{eq.map}), for each regular node $i$, states of the dynamics system (\ref{eq.dyn}) are mapped and partitioned into two sub-states $z_{\mathcal{D}_i}[k]$ and $z_{\mathcal{U}_i}[k]$ corresponding to the detectable and undetectable eigenvalues of the node $i$, respectively. Using the designed Luenberger observers, $\hat{z}_{D_i}[k]$ converges to $z_{D_i}[k]$ asymptotically. As an SR-MEDAG exists for each $\lambda_j \in \sigma_U(A)$, the result of Lemma~\ref{lm.netest} also holds. Consequently, node $i$ is able to estimate $z_{\mathcal{U}_i}[k]$ asymptotically even in the presence of communication delays, asynchrony and adversarial actions of smart spoofers. Combining these results, we infer that node $i$ can estimate the entire state $z[k]$ which leads to resiliently observing $x[k]$ using the transformation $x[k]=\Psi z[k]$. This completes the proof.
\end{proof}

\subsection{Spoof-Resilient Graph Topologies}
In this subsection, we characterize graph topologies  which ensures termination of the SR-MEDAG construction phase for each $\lambda_j \in \sigma_U (A)$ under misbehavior of smart spoofers.

\begin{lemma} \rm \label{lm.gterm}
The SR-MEDAG construction phase terminates for $\lambda_j \in \sigma_U (A)$ if $\mathcal{G}$ is strongly $ \big( 3(\beta +1)f+1 \big)$ -robust w.r.t. $\mathcal{S}_j$.
\end{lemma}

\begin{proof}
Contradiction is used for the proof. Consider any $\lambda_j \in \sigma_U (A)$ and let $\mathcal{G}$ be strongly $ \big( 3(\beta+1)f+1 \big)$-robust w.r.t. the set of source nodes $\mathcal{S}_j$. If the SR-MEDAG construction phase does not terminate for $\lambda_j$, there exists a set of regular nodes $\mathcal{C} \subseteq \mathcal{V} \setminus \mathcal{S}_j$ which never update their counter values $c_i(j)$ from 0 to 1 for $i \in \mathcal{C}$. As $\mathcal{G}$ is strongly $ \big( 3(\beta+1)f+1 \big)$-robust w.r.t. $\mathcal{S}_j$, it follows that $\mathcal{C}$ is $ \big( 3(\beta+1)f+1 \big)$-reachable, i.e., there exists a node $i \in \mathcal{C}$ which has at least $3(\beta+1)f+1$ neighbors outside $\mathcal{C}$. Under the $f$-local smart spoofer model, at most $f$ of these nodes are smart spoofers which are able either to misbehave themselves or to impersonate $\beta f$ regular nodes during the SR-MEDAG construction phase. So, at least, $ 2(\beta+1)f+1$ of them are regular nodes with $c_i(j) = 1$ which must have transmitted $\chi_j$ to node $i$. Thus, node $i$ must have changed $c_i(j)$ from 0 to 1 at some point of time, according to the rules of Algorithm~\ref{alg1}. This is a contradiction.
\end{proof}

\begin{proposition} \rm \label{prop.Kj1}
Suppose that $\mathcal{G}$ is strongly $ \big( 3(\beta+1)f+1 \big)$-robust w.r.t. $\mathcal{S}_j, \ \forall \lambda_j \in \sigma_U (A)$, and let the SR-MEDAG construction phase starts at $k=0$. Then, $\bar{K}_j$ in Algorithm~\ref{alg1} is upper bounded by $\bar{l}_j \big( (\eta +1) \bar{k} + \bar{\tau}+1 \big)$ where $\eta = \beta \lfloor (\bar{\tau}-\bar{k})/\bar{k} \rfloor$ and $\bar{l}$ is length of the longest path of $\mathcal{G}^j$.
\end{proposition}

\begin{proof}
Since $\mathcal{G}$ is strongly $ \big( 3(\beta+1)f+1 \big)$-robust w.r.t. $\mathcal{S}_j$, according to Lemma~\ref{lm.gterm}, each regular node has at least $2(\beta+1)f+1$ parent nodes which remain safe from spoofing and transmit $\chi_j$ at least once to the regular node. Each of these parent nodes has to make at least an update within $\bar{k}$ consecutive steps. 

We consider two separate cases: i) $\bar{\tau} \leq \bar{k}$ and ii) $\bar{\tau}>\bar{k}$. Let node $i$ be in $\mathcal{L}_1^j$ of $\mathcal{G}^j$. In the first case, each smart spoofer $s \in \mathcal{N}_i^j$ is able to impersonate at most $\beta f$ parent nodes of node $i$ in each consecutive $\bar{k}$ steps. So, the other $2(\beta+1)f+1$ parent nodes will remain safe and transmit $\chi_j$ to node $i$. For the second case, we consider the worst case where all these parent nodes postpone their updates by $\bar{k}-1$ steps and communicate to their neighbors with $\bar{\tau}$ steps delay (because we are seeking the maximum time steps that a smart spoofer can prevent exactly $2(\beta+1)f+1$ parent nodes to transmit $\chi_j$ to their neighbors). As $\bar{\tau}>\bar{k}$, the smart spoofer has more $\eta = \beta \lfloor (\bar{\tau}-\bar{k})/\bar{k} \rfloor$ capacity after the first update to impersonate more parent nodes. However, they cannot do it permanently, i.e. at some time the parent nodes will transmit $\chi_j$ to node $i$. Suppose that the smart spoofer use this additional capacity to impersonate just one additional parent node. In fact, $\eta$ updates of this parent node will be spoofed which takes $\eta \bar{k}$ time steps. Considering the first $\bar{k}-1$ steps that the parent node may postpone its first update and $\bar{\tau}$ steps delay of its last update, the overall time that the spoofed parent node succeeds to transmit $\chi_j$ to node $i$ will be $(\eta +1) \bar{k} + \bar{\tau} +1$ steps. 

Now, consider the last node in the longest path of $\mathcal{G}^j$ which is the last node that updates its counter value from 0 to 1. Let the length of the longest path of $\mathcal{G}^j$ be $\bar{l}_j$ and each node in this path updates its counter value from 0 to 1 after at most $(\eta +1) \bar{k} + \bar{\tau} +1$. Then the maximum time needed for each regular node to keep transmitting $\chi_j$ is bounded by $\bar{l}_j \big( (\eta +1) \bar{k} + \bar{\tau}+1 \big)$.
\end{proof}

\begin{remark} \rm
The maximum time that is needed for each regular node $i \in \mathcal{L}_m^j$ to keep transmitting $\chi_j$ is bounded by $(\bar{l}_j-m) \big( (\eta +1) \bar{k} + \bar{\tau}+1 \big)$. However, as the regular nodes cannot characterize their layer numbers, they have to keep transmitting the message $\chi_j$ up to $\bar{K}_j$ steps. 
\end{remark}

Finally, we propose the overall constraint on the network topology which makes sure that the network achieves resilient omniscience despite of smart spoofing actions.

\begin{theorem} \rm \label{th.over}
Resilient omniscience of a network with communication delays and asynchrony under $f$-local smart spoofer model is achieved using the proposed estimation scheme if $\mathcal{G}$ is strongly $ \big( 3(\beta+1)f+1 \big)$-robust w.r.t. $\mathcal{S}_j, \ \forall \lambda_j \in \sigma_U (A)$.
\end{theorem}

\begin{proof}
According to Lemma~\ref{lm.gterm}, the SR-MEDAG construction phase terminates for every undetectable eigenvalue $\lambda_j$ if $\mathcal{G}$ is strongly $ \big( 3(\beta+1)f+1 \big)$-robust w.r.t. $\mathcal{S}_j, \ \forall \lambda_j \in \sigma_U (A)$. Thus, based on Theorem~\ref{th.SRMEDAG}, as SR-MEDAG exists for every $\lambda_j \in \sigma_U(A)$. Finally, from Theorem~\ref{th.est}, the existence of an SR-MEDAG for every $\lambda_j \in \sigma_U(A)$ leads to resilient omniscience by using our proposed distributed estimation scheme in a network with communication delays and asynchrony under $f$-local smart spoofer model.
\end{proof}

\begin{remark} \rm
Suppose that smart spoofers impersonate none of the regular nodes during SR-MEDAG construction phase. Then, the sufficient constraint on the network topology to achieve resilient omniscience is strongly $ \big( 2(\beta+1)f+1 \big)$-robust w.r.t. $\mathcal{S}_j, \ \forall \lambda_j \in \sigma_U (A)$.
\end{remark}

Note that the presented sufficient conditions on the topology will be the same as the case of Byzantine attacks, proposed in \cite{mitra2018resilient}, if we set the parameter $\beta=0$. It means that the estimation will converge for all regular nodes under a simpler topology, i.e. strongly ($3f+1$)-robust w.r.t. $\mathcal{S}_j, \ \forall \lambda_j \in \sigma_U (A)$. This is consistent with the most important massage of our paper which asserts that asynchronous networks are more susceptible against cyber attacks; asynchronous networks can be threaten by adversaries that are stronger than Byzantine nodes, i.e. smart spoofers, which can use free time-steps between updates of regular nodes to impersonate some of them in order to mislead the others.

\subsection{Time-Varying Networks}
In the presented results so far, the observers over network $\mathcal{G}$  were supposed to be fixed, that is, the edge set $\mathcal{E}$ was time invariant. We now reconsider the results with a partially asynchronous time-varying network $\mathcal{G}[k] = (\mathcal{V}, \mathcal{E}[k])$ instead of the original time-invariant graph $\mathcal{G}$ earlier. To this end, similar to what is presented in \cite{dibaji2017automatica}, we define a jointly graph robustness measure as follows.

\begin{definition} \rm \label{df.joint}
(Jointly strongly $r$-robust w.r.t. $\mathcal{S}$) The time-varying graph $\mathcal{G}[k] = (\mathcal{V}, \mathcal{E}[k])$ is said to be jointly strongly $r$-robust w.r.t. $\mathcal{S}$ if there exists a fixed $\bar{\mu} \geq 0$ such that $\bigcup_{\mu=0}^{\bar{\mu}} \mathcal{G}[k-\mu]$, $k \in \mathbb{Z}_{\geq \bar{\mu}}$, is strongly $r$-robust w.r.t. $\mathcal{S}$.
\end{definition}

Referring to Lemma~\ref{lm.kscap}, the capacity of smart spoofers for impersonating regular nodes is bounded within each consecutive $\bar{k}$ steps by $\beta$. Thus, the horizon parameter $\bar{\mu}$ of time-varying graph $\mathcal{G}[k]$ has to satisfy the following inequality:
\begin{equation} \label{eq.joint}
\bar{\mu} \leq \bar{k}.
\end{equation}
Note that, otherwise, each smart spoofer would have extra capacity to impersonate more than $\beta$ regular nodes after each $\bar{k}$ steps. 

Now, the following result states the extension of our main result (Theorem~\ref{th.over}) for the case of time-varying networks.

\begin{corollary} \rm
Resilient omniscience of a network with communication delays and asynchrony under $f$-local smart spoofer model can be achieved using the proposed estimation scheme if $\mathcal{G}$ is jointly strongly $ \big( 3(\beta+1)f+1 \big)$-robust w.r.t. $\mathcal{S}_j, \ \forall \lambda_j \in \sigma_U (A)$, under condition (\ref{eq.joint}).
\end{corollary}

Similarly, reconsidering the case where smart spoofers do not impersonate any of the regular nodes during the SR-MEDAG construction phase, the following result holds for time-varying networks.

\begin{corollary} \rm
Suppose that smart spoofers impersonate none of the regular nodes during SR-MEDAG construction phase. Then, the sufficient topology constraint on the network $\mathcal{G}[k]$ to achieve resilient omniscience is jointly strongly $ \big( 2(\beta+1)f+1 \big)$-robust w.r.t. $\mathcal{S}_j, \ \forall \lambda_j \in \sigma_U (A)$, with condition (\ref{eq.joint}).
\end{corollary}

These results follow Lemmas~\ref{lm.sandwich} to \ref{lm.gterm}  as the time-invariant nature of the original graph $\mathcal{G}$ is not used in the proofs.\\

\subsection{Randomized Update Rule}
Consider the case that each regular node, at each time instant, randomly decides whether to update its state estimate or not. That is, the follower node $i \in \mathcal{R}$ updates its state estimate at each time instant $k$ with the probability of $P_i[k]$. Note that with such updates, the algorithm remains fully distributed. Even the probabilities $P_i[k]$ need not be identical. Intuitively, this is in alignment with Assumption~\ref{asum3} as the regular node will update at least once within each consecutive $\bar{k}$ steps. With this strategy, the topology constraint required for resilient omniscience can be relaxed. This is because the smart spoofers cannot predict the update times in advance and need to use more of their spoofing capacity to make sure that the regular nodes, at each time step, receive and accept false data with fake identities; they cannot impersonate other nodes in a systematic manner in each consecutive $\bar{k}$ steps. In fact, regular nodes utilized randomization in update times as a defensive means against smart spoofers.

What follows is the modification of Theorem~\ref{th.over} for the suggested network with randomized updating strategy.

\begin{figure}[t]
	\def \svgscale{.6}
	\hspace{.7cm}
	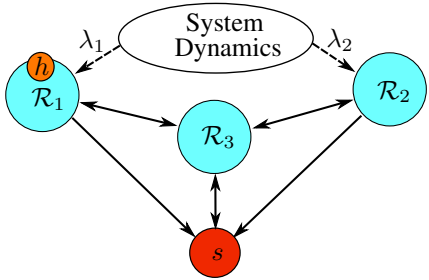
    \caption{A sample network that is $ \big( 2(\beta +1)f+1  \big)$-robust w.r.t. $\lambda_j$. Node $h \in \mathcal{R}_1$ is impersonated by the smart spoofer $s$ for the nodes in set $\mathcal{R}_3$.} \label{fig.sim} 
\end{figure} 

\begin{theorem} \rm \label{th.rand}
Resilient omniscience of a network with communication delays and asynchrony under the $f$-local smart spoofer model can be achieved using the proposed estimation scheme if each follower node $i \in \mathcal{R}$, at each time instant $k$, updates using rule (\ref{eq.consda}) with the probability of $P_i \in (0,1]$ and if $\mathcal{G}$ is strongly $ \big( 3(\beta^\prime +1)f+1 \big)$-robust w.r.t. $\mathcal{S}_j, \ \forall \lambda_j \in \sigma_U (A)$, where $\beta^\prime = \lfloor \beta / \bar{k} \rfloor +1$.
\end{theorem}

\begin{proof}
Referring to Lemma~\ref{lm.kscap}, each smart spoofer was able to impersonate at most $\beta$ regular nodes within $\bar{k}$ steps in case the smart spoofers knew when exactly each regular node updates its state estimate. Now, consider that each regular node $i$ make an update at each time instant $k$ randomly with a probability of $P_i[k] \in (0,1]$. Then, each smart spoofer $s \in \mathcal{N}_i$ cannot predict when exactly the node $i$ updates; so it has to impersonate the incoming neighbors of the node $i$ for all the time steps within each consecutive $\bar{k}$ steps, that is $\bar{k}$ times. As a result, the smart spoofers need to dedicate more capacity to produce faulty data packets with the mimicked identities of the neighbors of the node $i$. Thus, the smart spoofers will be able to impersonate $ \lfloor \beta / \bar{k} \rfloor$ regular nodes for any of $\bar{k}$ steps and one regular node for a limited number of time-steps, i.e. less than $\bar{k}$. In this situation, to ensure that the smart spoofers cannot impersonate any extra regular nodes, we define $\beta^\prime = \lfloor \beta / \bar{k} \rfloor +1$. Accordingly, similar to the proof of Theorem~\ref{th.over}, the required topology constraint for omniscience based on the parameter $\beta^\prime$ is strongly $ \big( 3(\beta^\prime +1)f+1 \big)$-robust w.r.t. $\mathcal{S}_j, \ \forall \lambda_j \in \sigma_U (A)$.
\end{proof}

\begin{remark} \rm
Based on Lemma~\ref{lm.kscap}, we have $\beta^\prime = \lfloor (\alpha \bar{k} -1)/ \bar{k} \rfloor$. On the other hand, we know $\bar{k} \geq 1$. Therefore, it is concluded that $\beta^\prime = \alpha - 1$ in the case of a synchronous network ($\bar{k}=1$) and $\beta^\prime = \alpha $ in an asynchronous network ($\bar{k} > 1$).
\end{remark}

 \section{Simulation Results}\label{Sect: Simulations}

In this section, we present a simulation example to demonstrate how a smart spoofer can misbehave and how it can be restrained in a given network of distributed observers. In particular, we show why the constraints on the network topology, proposed in Theorems~\ref{th.est} and \ref{th.over}, are critical for achieving resilient omniscience under $f$-local smart spoofer model in the presence of asynchronous communications and delays.

\begin{figure}[t] 
	\begin{center}
	\includegraphics[scale=.55]{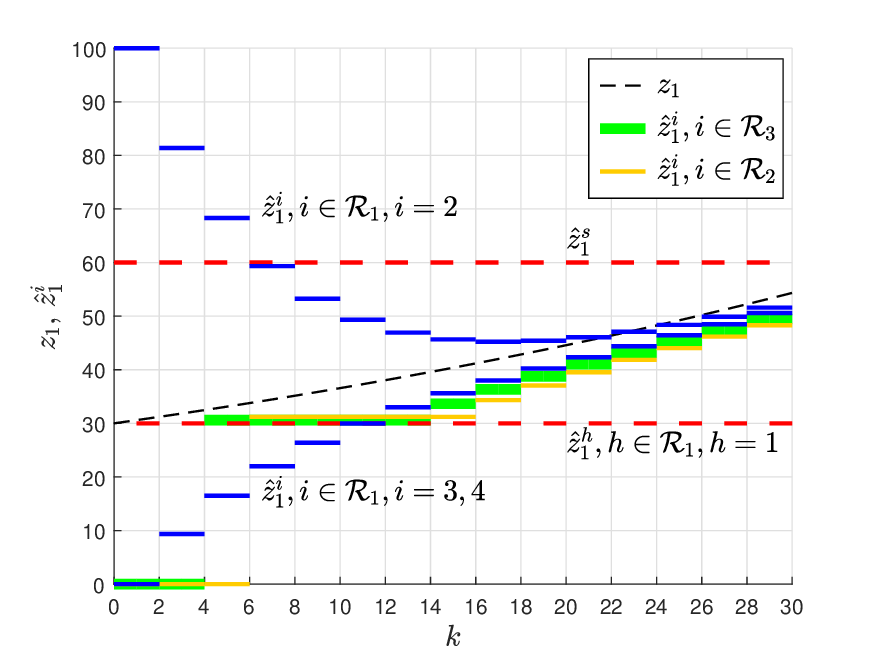}
	\caption{Omniscience of the sample network: the regular source nodes in $\mathcal{R}_1$ asymptotically estimate $z_1$ using Luenberger observers. The follower nodes in $\mathcal{R}_2$ and $\mathcal{R}_3$ can asymptotically estimate $z_1$ despite of the efforts smart spoofer does for misleading them by impersonating the node $h \in \mathcal{R}_1$. They also estimate $z_2$ since there is no spoofing for $\lambda_2$.}
	\label{fig.z1_omni}
	\end{center}
\end{figure}

\begin{figure}[t] 
	\begin{center}
	\includegraphics[scale=.55]{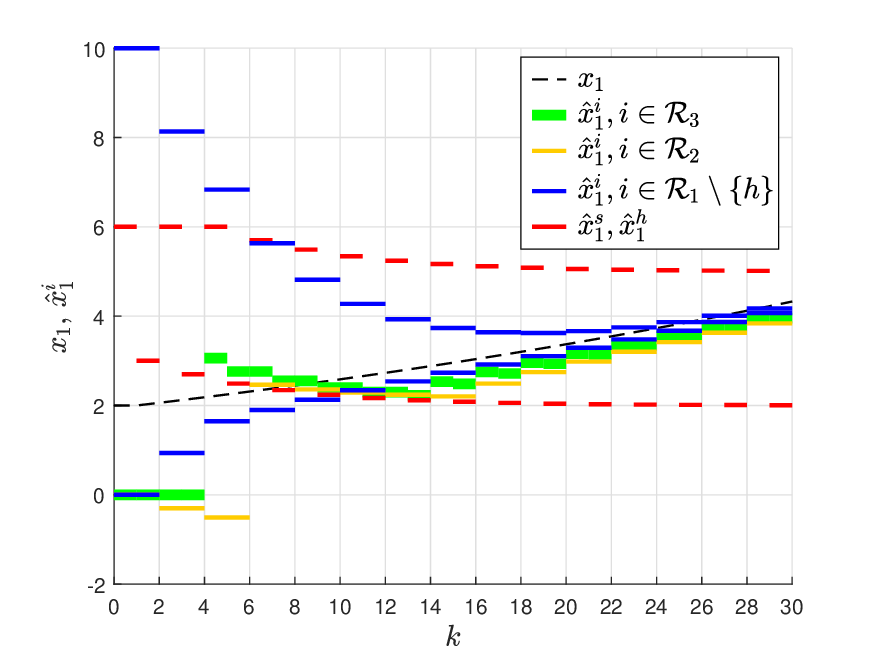}
	\caption{All the regular nodes truly estimate $x_1$ while the smart spoofer $s$ impersonates the node $h \in \mathcal{R}_1$ for the nodes in $\mathcal{R}_3$.}
	\label{fig.x1_omni}
	\end{center}
\end{figure}

\begin{figure}[t] 
	\begin{center}
	\includegraphics[scale=.55]{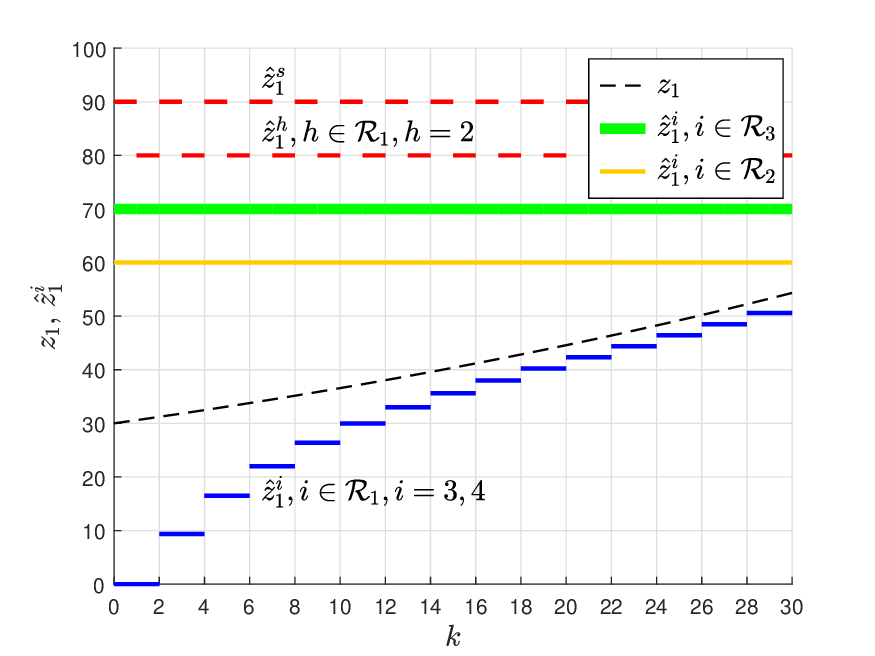}
	\caption{The smart spoofer $s$ prevents the follower nodes for $\lambda_1$ in $\mathcal{R}_2$ and $\mathcal{R}_3$ to estimate $z_1$ by sending false estimate values on behalf of node $h \in \mathcal{R}_1$ to all the nodes in $\mathcal{R}_3$.}
	\label{fig.z1_fail}
	\end{center}
\end{figure}

\begin{figure}[t] 
	\begin{center}
	\includegraphics[scale=.55]{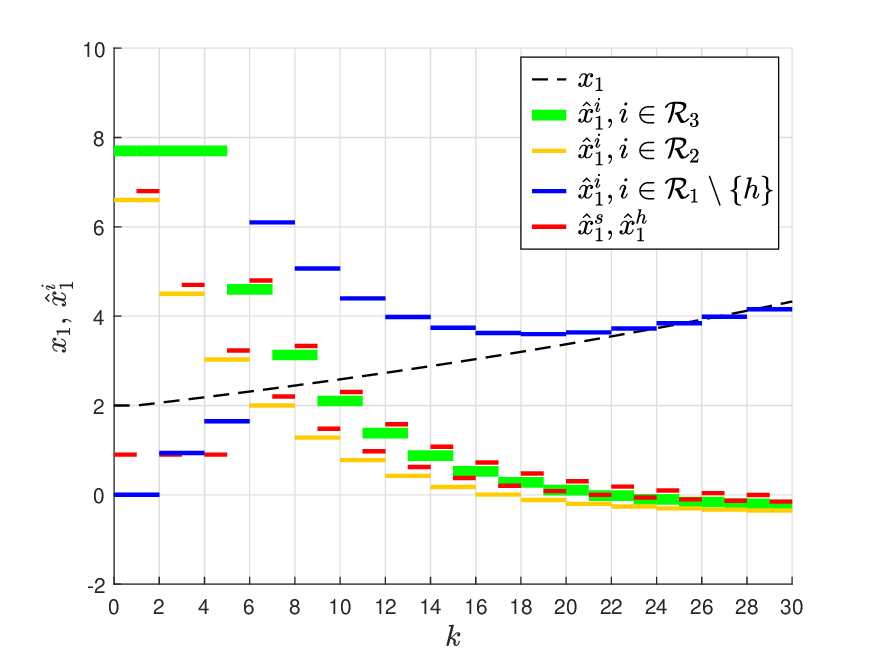}
	\caption{The regular nodes in $\mathcal{R}_2$ and $\mathcal{R}_3$ fail to estimate $x_1$ if the smart spoofer $s$ impersonates the node $h \in \mathcal{R}_1$ for the nodes in $\mathcal{R}_3$.}
	\label{fig.x1_fail}
	\end{center}
\end{figure}

To this end, consider the network illustrated in Fig.~\ref{fig.sim}. The directed edges of the graph represent all to one connections and edges pointing in both directions represent all to all connections. The network has three sets of regular nodes $\mathcal{R}_1$, $\mathcal{R}_2$, $\mathcal{R}_3$, and a smart spoofer $s$ ($f=1$). The capacity of $s$ is assumed to be $\alpha = 1$ and all of the nodes are supposed to make, at least, an update within $\bar{k}=2$ steps. Thus, the parameter is set as $\beta =1$. There are $2(\beta +1)f$ nodes in each of the sets $\mathcal{R}_1$ and $\mathcal{R}_2$ and $2(\beta +1)f+1$ nodes in the set $\mathcal{R}_3$, where the nodes within each set are not connected. Furthermore, communication delays over the network are defined as follows:

\begin{equation*}
\tau_{ij}=
\begin{cases}
3, & \mbox{if} \ (i \in \mathcal{R}_3 \ \mbox{and} \ j \in \mathcal{R}_2)\\
2, & \mbox{if} \ (i \in \mathcal{R}_2 \ \mbox{and} \ j \in \mathcal{R}_3) \\
& \mbox{or} \ (i \in \mathcal{R}_3 \ \mbox{and} \ j \in \mathcal{R}_1)\\
1, & \mbox{if} \ (i \in \mathcal{R}_1 \ \mbox{and} \ j \in \mathcal{R}_3) \\
& \mbox{or} \ (i \in \mathcal{R}_3 \ \mbox{and} \ j=s)\\
0, & \mbox{if} \ i=s \ \mbox{and} \ (j \in \mathcal{R}_1 \ \mbox{or} \ j \in \mathcal{R}_2)
\end{cases}
\end{equation*}

Note that the smart spoofer fully knows the dynamic system and the observation models of the regular nodes and calculates its own and the impersonated states estimations in a way that the targeting regular nodes fail to reach omniscience. Thus, to give a better intuition, we deal with the transformed dynamic system and states in our simulations; the original system can be analyzed accordingly. For sake of simplicity, we use the terms ``send" and ``receive" with $\hat{z}_1^i$ (or $\hat{z}_2^i$), while our purpose is the original state $\hat{x}_1^i$ associated to the $\hat{z}_1^i$.

The transition matrix of the original dynamic system and the initial state are assumed to be
\begin{equation*}
A=\left[
\begin{matrix}
    0.98  &  0.02\\
   -0.04  &  1.04
\end{matrix}
\right], \
x_0=\left[
\begin{matrix}
    2\\
    5
\end{matrix}
\right],
\end{equation*}
which, according to (\ref{eq.map}), are transformed by
\begin{equation*}
\Psi=\left[
\begin{matrix}
0.1 & 1\\
0.2 & 1
\end{matrix}
\right]
\end{equation*}
into the following diagonal system and initial state:
\begin{equation*}
\bar{A}=\left[
\begin{matrix}
1.02 & 0\\
0    & 1
\end{matrix}
\right], \
z_0=\left[
\begin{matrix}
    30\\
    -1
\end{matrix}
\right].
\end{equation*}
The first eigenvalue of the system ($\lambda_1=1.02$) is unstable and the second one ($\lambda_2=1$) is marginally stable. The observation model of the network system is assumed to be $C_i = [-10 \ 10]$, $ \forall i \in \mathcal{R}_1$, $C_i = [2 \ -1]$, $\forall i \in \mathcal{R}_2$ and $C_i = 0$, $\forall i \in \mathcal{R}_3$. The transformation of the observation model is given by $\Psi$ as $\bar{C}_i = [1 \ 0]$, $\forall i \in \mathcal{R}_1$, $\bar{C}_i = [0 \ 1]$, $\forall i \in \mathcal{R}_2$ and $\bar{C}_i = 0$, $\forall i \in \mathcal{R}_3$. This means that the nodes in $\mathcal{R}_1$ are source nodes for $\lambda_1$ and followers for $\lambda_2$ as they can only detect $\lambda_1$, the nodes in $\mathcal{R}_2$ are source nodes for $\lambda_2$ and followers for $\lambda_1$, and the nodes in $\mathcal{R}_3$ are followers for both $\lambda_1$ and $\lambda_2$. Also, the nodes in the set $\mathcal{R}_3$ and the node $s$ update at all time instants and the nodes in $\mathcal{R}_1$ and $\mathcal{R}_2$ update at time instants $k=m\bar{k}, \ m \in \mathbb{Z}_+$.

We present the simulation results in two test scenarios. In both scenarios, the smart spoofer just impersonate only one node in $\mathcal{R}_1$ and only for the state $z_1$. Thus, all the nodes of the network will accurately estimate the state $z_2$. In scenario 1, we show that the follower nodes for $\lambda_1$ can asymptotically estimate $z_1$ even though the smart spoofer $s$ tries to mislead the follower nodes but the network finally achieves omniscience. However, in scenario 2, the follower nodes cannot reach a true estimate of $z_1$ as the required topology constraint for estimation (Theorem~\ref{th.est}) is not satisfied. 

\begin{scenario} \rm The smart spoofer sends the message $\chi_s=1$ to all the nodes in $\mathcal{R}_3$ to pretend that it is a parent node for $\lambda_1$. Although the smart spoofer can impersonate a regular node during the SR-MEDAG construction phase, it decides not to do so and goes through the estimation phase. The initial estimates of $z_1$ for nodes in $\mathcal{R}_1$ are $\hat{z}_1^i[0]=100$, $i=1,2$ and $\hat{z}_1^i[0]=0$, $i=3,4$. Also, we have $\hat{z}_1^i[0] = 0$, $i \in \mathcal{R}_2$, and $\hat{z}_1^i[0] = 0, i \in \mathcal{R}_3$. Moreover, each regular node $i \in \mathcal{R}_1$ uses a Luenberger observer with the gain $L_i^1 = 0.5$ to estimate $z_1$. Starting the estimation phase, the smart spoofer $s$ sends two sequences of estimate values to all the nodes in $\mathcal{R}_3$ in a way that the receiving estimate values from the nodes in $\mathcal{R}_1$ are eliminated in local filtering: i) the estimate value $\hat{z}_1^s[k]=60$, where $k = m\bar{k}-1, m \in \mathbb{Z}_+$, which keeps the smart spoofer $s$ among the accepted neighbors of the nodes in $\mathcal{R}_3$ as each node has to send a data packet at least in each consecutive $\bar{k}$ steps, ii) a false estimate value $\hat{z}_1^h[k]=30$, where $k = m\bar{k}, m \in \mathbb{Z}_{\geq 0}$, on behalf of the node $h=1, h \in \mathcal{R}_1$ to the nodes in $\mathcal{R}_3$. As shown in Fig.~\ref{fig.z1_omni}, all the regular nodes can estimate $z_1$ although the smart spoofer caused a deviation in the estimations of the nodes in $\mathcal{R}_3$ (and accordingly the nodes in $\mathcal{R}_2$) up to time instant $k=12$. Note that the estimate value of $\hat{z}_1^h[k]$ will not be filtered only if it converges to $z_1$. In fact, referring to Lemma~\ref{lm.sandwich}, estimate values of the smart spoofer are sandwiched by estimate values of the regular parent nodes, thanks to our proposed rule (\ref{eq.consda}) and the network topology constraint discussed in Theorem~\ref{th.est}.
\end{scenario}

\begin{scenario} \rm Here, the smart spoofer $s$ sends a message $\chi_s=1$ to the nodes in $\mathcal{R}_3$ while it impersonates the node $p=1, p\in \mathcal{R}_1$ in the SR-MEDAG construction phase by setting the message $\chi_p = 0$, i.e. the node $p$ cannot be a parent node of the nodes in $\mathcal{R}_3$ for $\lambda_1$. In other words, Algorithm~\ref{alg1} does not terminate in the case of $\lambda_1$ for the nodes in $\mathcal{R}_3$. In fact, the constraint on the network topology is not satisfied for the estimation of $z_1$ since the nodes in $\mathcal{R}_3$ recognizes only $2(\beta +1)f$ parent nodes for $\lambda_1$. However, assume that the nodes in $\mathcal{R}_3$ decide to start the estimation regardless of the termination of the SR-MEDAG construction phase. As a result, the smart spoofer is able to impersonate one more regular node of the set $\mathcal{R}_1$ this time in the estimation phase. The initial estimates of $z_1$ are given as $\hat{z}_1^i[0]=10, i \in \mathcal{R}_1, i=2$, $\hat{z}_1^i[0]=0$, $i \in \mathcal{R}_1, i=3,4$, $\hat{z}_1^i[0] = 6$, $i \in \mathcal{R}_2$, and $\hat{z}_1^i[0] = 7, i \in \mathcal{R}_3$. Again, the smart spoofer $s$ sends two sequences of estimate values to the nodes in $\mathcal{R}_3$ in a way that the estimate values of the nodes $i \in \mathcal{R}_1$, $i=3,4$, are eliminated in local filtering: i) the estimate value $\hat{z}_1^s[k]=8$, where $k = m\bar{k}-1, m \in \mathbb{Z}_+$, which keeps the smart spoofer $s$ among the accepted neighbors of the nodes in $\mathcal{R}_3$, ii) a false estimate value $\hat{z}_1^h[k]=9$, where $k = m\bar{k}, m \in \mathbb{Z}_{\geq 0}$, on behalf of the node $h=1, h \in \mathcal{R}_1$, to the nodes in $\mathcal{R}_3$. Fig.~\ref{fig.z1_fail} shows the consequence of spoofing in the estimations. The initial estimate values of the nodes in $\mathcal{R}_2$ and $\mathcal{R}_3$ remain constant for all the future time. It is noteworthy that, not only the nodes in $\mathcal{R}_3$ are affected by the spoofing, but the nodes in $\mathcal{R}_2$ are also affected indirectly and none of them can reach omniscience for $z_1$.  
\end{scenario}

While we primarily analyzed the success or failure of the network omniscience in the transformed dynamic system, the main results are valid for the original dynamic system with different time histories of state values (Fig.~\ref{fig.x1_omni} and Fig.~\ref{fig.x1_fail}).

\section{Conclusion}\label{Sect: Conclusion}
Combining Byzantine adversarial model and spoofing as a misbehaving technique, we introduced a new type of cyber attack called smart spoofing. Then, we investigated the problem of distributed observer design for LTI systems in the presence of this attack which uses the asynchrony in communications to threaten the network. Using a two-step distributed mechanism, including a pre-executing algorithm for recognizing the trusted neighbors and a local-filtering algorithm for removing possible incorrect values induced by the adversarial nodes, the regular nodes can achieve resilient observation over so-called strongly robust graphs. We proposed resilient topology constraints on static and time-varying networks under the proposed adversarial threat. Numerical simulations with a sample network validate our analytic results. The proposed designs are applicable to a vast range of networked systems. In future studies, we consider resilient consensus problems prone to the smart spoofing attacks.

\section*{Acknowledgements}

I offer my sincerest gratitude to Dr. Seyed Mehran Dibaji and Prof. Hideaki Ishii for the time they dedicated to me for useful discussions on the topic as well as their technical comments which significantly helped me to improve the quality of this paper.

\bibliography{filterSpoofing}


\end{document}